\documentclass[11pt]{article}

\usepackage{graphicx} 
\usepackage{amssymb,amsmath,amsthm,csquotes,mathtools,amssymb}
\usepackage{bbm}
\usepackage[backend=biber,style=authoryear-comp,citestyle=authoryear-comp,maxnames=1]{biblatex}
\usepackage{mathtools}
\usepackage{csquotes}
\usepackage{amsthm}
\usepackage{amsmath}
\usepackage{mathabx}
\usepackage{dsfont}
\usepackage{enumerate}
\usepackage{biblatex}
\usepackage{hyperref}
\usepackage{physics}
\usepackage{multirow,array}
\usepackage{xcolor}
\newtheorem{claim}{Claim}
\usepackage{subcaption}
\usepackage[margin=1in]{geometry}
\usepackage{tikz}
\usetikzlibrary{positioning, arrows.meta}

\newtheorem{definition}{Definition}
\newtheorem{theorem}{Theorem}

\newtheorem*{lemma}{Lemma}

\title{Quantum Primitive for Output-Hiding Function Sharing}

\author{Olivia R. Hartzell\thanks{I thank Terry Rudolph, David Kagan, and Jacob A. Barandes for helpful comments. \\
This work is the subject of U.S. Provisional Patent Application no. 64/045,601.\\
\copyright 2026 Olivia Hartzell. This work is licensed under a Creative Commons Attribution-NonCommercial-NoDerivatives 4.0 International License (CC BY-NC-ND 4.0) available at https://creativecommons.org/licenses/by-nc-nd/4.0/. Any commercial reuse or creation of derivative works without the express written consent of the author is prohibited.}}

\date{\today}

\begin{document}
\maketitle

\begin{abstract}
A quantum information–theoretic primitive is introduced for determining a discrete-valued function that depends on multiple parties’ local private inputs.

The primitive permits the parties to mutually learn each others' local inputs, and thereby determine function values, while their individual systems remain independent of these inputs. The resulting function values are shared among the parties, but may remain information-theoretically hidden from any external observer, as well as from adversarial state-preparation or measurement processes within the quantum system, in every iteration. In particular, while classically producing a shared function with these information-theoretic properties requires the use of private keys or hidden randomness, in the proposed setting it is achieved using quantum resources alone.

I outline the primitive's general properties while applications across a broad range of secure quantum communication and computation settings including; quantum key distribution, multi-party coordination and decision schemes, function evaluation, and in some settings, protocols for fairly generated private coins, are relegated to further publications.
\end{abstract}

\section{Introduction and Overview}
Let there be two space-like separated parties, Alice and Bob, indexed by
\[
i \in \{a, b\}
\]

Each party $i$ may hold local private information denoted by $x_i$, drawn from a domain $\mathcal{X}_i$, which may consist of real numbers, complex amplitudes, symbolic parameters, or any other form of information relevant to the implementation. For example, $x_i$ may denote the outcome of a local random coin flip, or information about party $i$'s current location. 

We define the joint domain as
\[
\mathcal{X} =  \mathcal{X}_a \times \mathcal{X}_b .
\]

The parties wish to compute the value of a function
\[ 
s : \mathcal{X} \rightarrow \mathcal{A},
\]
where the range $\mathcal{A}$ is a finite discrete alphabet,
\[
\mathcal{A} = \{ a_0, a_1, \ldots, a_{q-1} \},
\]
whose elements may represent numerical values, logical outcomes, symbols, or other distinguishable states. For convenience, we will focus on 
 \(\mathcal{A}=\mathbb{Z}_q\ = \{0,1,\ldots,q-1\}\) denoting a finite alphabet with cardinality $q$, but the construction is not limited to numeric outputs.

Thus, the function
\[
s(x_a, x_b)
\]

maps the parties' collective local inputs to one of a finite set of discrete outcomes.

For example, in a QKD implementation, parties may use this primitive over a number of iterations to generate a shared string of values, or a key, \(S = (s_1,s_2,\dots,s_m)\) where each  $s \in \{0, 1\}\ $ for all $s \in S$. $s$ may also represent the answer to a discrete-valued question such as; $s = \mathbbm{1}\{x_a > x_b\}$. 

\hfill

Parties are mutually trusted in that they are assumed to not collude with any attacker or external party. While they trust their own local operations, they may not trust resources outside of their own local systems.

\subsection{Primitive Properties}
The primitive allows the parties to achieve the following: 
\begin{enumerate}
    \item Generate a single measurement event that depends on the parties' joint information $(x_a, x_b)$, while the parties' local systems are independent of their individual inputs. Given this measurement event, parties' recover each others' local inputs $(x_a, x_b)$ with certainty in an ideal execution environment. 

    \item Given this measurement event, parties can record corresponding function values $s(x_a, x_b)$, such that the public measurement event is independent of $s(x_a, x_b)$; and therefore its value remains information-theoretically hidden from those external to the protocol.

    \item When the parties' both compute the value $s(x_a, x_b)$ correctly, the parties' recorded values remain information-theoretically hidden from adversaries \textit{within} the quantum system (including adversarial state-preparation or measurement attacks), in \textit{every iteration} of the primitive.  
\end{enumerate}

The parties therefore generate a shared value $s(x_a, x_b)$ which remains information-theoretically hidden both from those external to the protocol, and from those within the quantum system itself. In particular, this is achieved with quantum resources alone, \textit{without} any additional private keys, hidden randomness, or classical communication -- a capability that is provably classically impossible. 
In this sense, the primitive permits the parties to reach an agreement $s(x_a, x_b)$, while any physical record left behind is fully independent of the agreement itself. 

\subsection{Related Work}
The information-sharing process described above in $1.$ is similar in spirit to quantum superdense coding (\cite{bennett1992communication}) but where \textit{multiple parties} hold private information, as in \cite{yang2020bidirectional, pan2024evolution}. In this particular setting, the primitive permits both parties to mutually learn each others' local information with certainty, with a single measurement event, in contrast to \cite{nguyen2004quantum} in which parties' share information sequentially. Additionally, the process outlined that follows is more general than that described in \cite{yang2020bidirectional}, which restricts parties' local operations to the generalized Pauli operators. In this setting, the information-sharing and resulting security properties arise from the interaction between the selected local unitary operators and the chosen joint measurement basis. As a result, the underlying mechanism is not tied to the specific algebraic structure of the generalized Pauli operators and applies to a broader class of encoding schemes.

In addition, the parties may use their shared local information to generate shared knowledge, or function value $s(x_a, x_b)$ with the information-theoretic properties previously described. In particular, this capability is not a classical \textit{speed-up} for computing $s(x_a, x_b)$. Importantly, it bears a fundamentally different information-sharing property from secure multi-party computation (MPC). In the proposed primitive, each party \textit{learns} the others' local information, while their output $s(x_a, x_b)$ remains information-theoretically hidden. This structure contrasts sharply with classical and quantum secure function evaluation protocols (\cite{Yao1982, GoldreichMicaliWigderson1987, CrepeauGottesmanSmith2002}), including quantum secure multi-party summation protocols (\cite{Shi2016}, \cite{Lu2024}, \cite{wu2023quantum}, \cite{yi2021quantum}), in which the parties wish to compute a joint function value, while \textit{not} revealing their inputs to one another. In particular, in these settings, the resulting function value may be public. 

In the proposed setting, parties may compute $s(x_a, x_b)$ which depends on \textit{both} of their inputs, while $s(x_a, x_b)$ remains information-theoretically hidden from others, \textit{without} the use of private keys or hidden randomness. Notably, this structure differs from other measurement device independent quantum dialogue settings which require QKD as a \textit{pre-requisite} for security (\cite{maitra2017measurement}).

A more detailed discussion of QKD implementations and how the proposed settings differs from standard settings (in particular the security properties of $s(x_a, x_b)$), are relegated to an accompanying write-up along with other applications for shared decisions $s(x_a, x_b)$. 

A second additional write-up discusses adaptations when parties' may have strategic preferences over $s(x_a, x_b)$, and therefore may not be mutually-trusted. In these settings, under certain assumptions, the primitive may be used to generate a private unbiased coin.

\section{Procedure}\label{sec:procedure}
First, the parties agree on the objects that define the primitive: an initial quantum state, a finite set of local unitary operators, a measurement basis, and a recording strategy. This agreement may be public. Once these choices are mutually agreed upon, the parties may separate and need not communicate during the execution of the primitive.

The parties share access to an initial joint quantum state
\[
    \ket{\psi} \in \mathcal{H} = \bigotimes_{i=a}^{b} \mathcal{H}_i,
\]
where each  \( \mathcal{H}_i \) is a finite-dimensional Hilbert space.  
While the parties choose this initial state \( \ket{\psi} \) for their particular goals, the state may not have been prepared by the parties themselves -- they may be given access to this initial state by an untrusted party, who will be referred to as Charley. Each party \( i \) has access only to their local subsystem, i.e., to the degrees of freedom 
corresponding to \( \mathcal{H}_i \).

Additionally, each party \( i \) has access to a finite set of local unitaries,
\[
\mathcal{U}_i = \{\, U_i^{(k)} \,\}_{k=1}^{m_i} \subset \mathrm{U}(\mathcal{H}_i).
\]
where \(\mathrm{U}(\mathcal{H}_i)\) denotes the set of all unitary operators acting on \(\mathcal{H}_i\).

The set of all possible joint unitaries implementable by the parties is given by
\[
\mathcal{U}_{a, b} =
\left\{
\, \bigotimes_{i=a}^{b} U_i^{(k_i)}
\;\middle|\;
U_i^{(k_i)} \in \mathcal{U}_i
\right\}.
\]

For each iteration of the primitive, party \(i\) realizes a local value
\[
x_i \in \mathcal{X}_i^{\mathrm{sub}} \subseteq \mathcal{X}_i,
\]
$\mathcal{X}_i^{\mathrm{sub}}$ is finite and discrete and denotes the subset of admissible values for a given iteration of the primitive.  
The ex-ante distribution over the subset
\[
\mathcal{X}^{\mathrm{sub}} = \mathcal{X}_a^{\mathrm{sub}} \times \mathcal{X}_b^{\mathrm{sub}}
\]
is assumed to be uniform. For example, these local values $x_i$ may be the results of local uniform-random coin flips. For convenience we will denote $X_i$ as the corresponding random variable uniformly distributed over $\mathcal{X}_i^{\mathrm{sub}} $.  

Given their realizations \( x_i \), each party selects a corresponding unitary operator \( U_i^{(x_i)} \in \mathcal{U}_i \).  
The mapping \( x_i \mapsto U_i^{(x_i)} \) is assumed to be bijective—that is, each admissible value of \( x_i \) uniquely determines a feasible unitary, and each unitary in \( \mathcal{U}_i \) corresponds to exactly one value of \( x_i \).\footnote{When parties may wish to share only \textit{subsets} of information, the mapping need not be bijective over the full domain. Instead, it suffices that it is bijective when restricted to a subset $\{x_i^1, x_i^2, \dots\}$, provided that the remaining properties of the primitive—particularly those governing the distribution of $s$ conditional on the shared inputs—continue to hold.}

After these unitaries 
are applied, the resulting global state becomes
\[
    \ket{\phi} = \left( \bigotimes_{i=a}^{b} U_i \right) \ket{\psi}.
\]

The state \( \ket{\phi} \) is then sent to a referee, or measurement station, who will be referred to as Eve, 
for measurement in an orthonormal basis, $\mathcal{J}$ of \(\bigotimes_{i=a}^{b} \mathcal{H}_i\) where \(\{\ket{j}\}_{j\in \mathcal{J}}\) will denote the corresponding projective measurement outcomes. 

Eve announces the measurement outcome to the parties. For simplicity, we will use \(j \in \mathcal{J}\) to denote both the measurement outcome label and the corresponding announced result, 
with \(\ket{j}\) referring explicitly to the basis state associated with outcome \(j\).

Importantly, Eve need not be trusted, and she may collude with the state-preparer, Charley. A detailed discussion on adversarial attacks and the security bounds under these attacks are relegated to section \ref{sec:adversarial_success_bounds}.

Given their chosen local unitary $U_i$, (or equivalently $x_i$) and the announcement $j$, each party $i$ records a value 
\[
s_i \in \mathcal{A} \cup \{\emptyset\}
\]

according to a pre-agreed recording strategy (for example, a coding rule in the QKD implementation).  This recording strategy must follow the properties defined in detail in Section \ref{sec:properties:recording-s}.

A schematic of the primitive procedure is depicted in Figure \ref{fig:quantum_protocol_ab_charley}.

\begin{figure}[ht]
\centering
\resizebox{\linewidth}{!}{%
\begin{tikzpicture}[
    node distance=2cm and 3cm,
    box/.style={
        draw,
        rounded corners,
        align=center,
        minimum height=1.5cm,
        minimum width=5cm,
        font=\large
    },
    arrow/.style={-Stealth, thick, line width=1pt}
]

\node[box] (agreement) {Alice and Bob agree on:\\
Initial state $\ket{\psi}$,\\
Sets of unitaries $\{U_A\}, \{U_B\}$,\\
Measurement basis $\mathcal{J}$,\\
Recording strategy $s$ (depends on $U_A$, $U_B$, $j$)
};

\node[box, below=2cm of agreement] (Charley) {Charley prepares $\ket{\psi}$};

\node[box, below left=of Charley, xshift=-2cm] (Alice) 
{Alice\\Locally generate $x_A$, choose $U_A$ via bijection \\apply $U_A$ to subsystem};

\node[box, below right=of Charley, xshift=2cm] (Bob) 
{Bob\\Locally generate $x_B$, choose $U_B$ via bijection \\apply $U_B$ to subsystem};

\node[box, below=4cm of Charley] (Eve) 
{Eve measures $(U_A \otimes U_B) \ket{\psi}$ in basis $\mathcal{J}$\\Announces measurement $j$};

\node[box, below=of Eve] (Record) 
{Alice and Bob record $\tilde{s}$ \\ using $U_i$, $j$, and pre-agreed strategy};

\draw[arrow] (agreement.south) -- (Charley.north);
\draw[arrow] (Charley.south) -- (Alice.north);
\draw[arrow] (Charley.south) -- (Bob.north);

\draw[arrow] (Alice.south) -- (Eve.north);
\draw[arrow] (Bob.south) -- (Eve.north);

\draw[arrow, dashed] (Eve.south) -- (Record.north);

\end{tikzpicture}%
}
\caption{Schematic of the quantum primitive for two parties (Alice and Bob).}
\label{fig:quantum_protocol_ab_charley}
\end{figure}
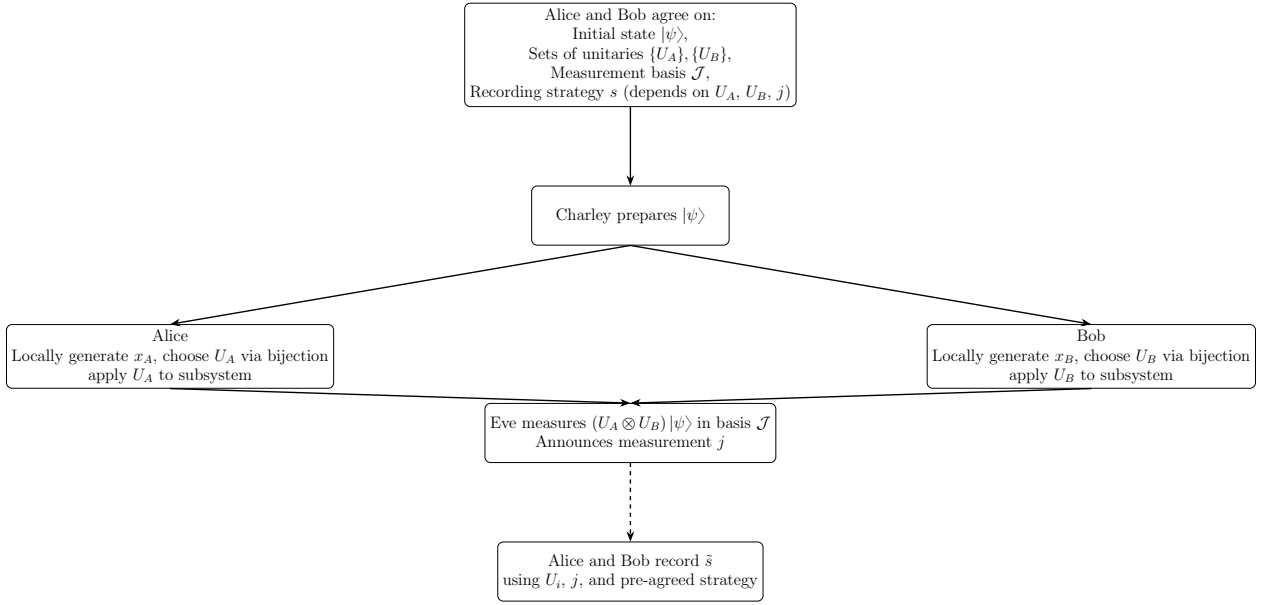

\section{Example Primitives}\label{sec:example_primitives}
I next provide examples of primitives that satisfy the information-revelation and security properties previously described. Various subsets of these primitives may be useful under various levels of trust in the initial state preparation and measurement. In other words, players may adapt the security guarantees by adjusting the subsets of feasible unitaries that they randomize over, or the recording strategies.

\subsection{Primary Example}\label{example-protocol-joint}

Consider the operator
\[
{\mathcal{J}} = \frac{1}{\sqrt{2}}
\begin{bmatrix}
1 & 0 & 0 & -\mathrm{i} \\
0 & 1 & -\mathrm{i} & 0 \\
0 & -\mathrm{i} & 1 & 0 \\
-\mathrm{i} & 0 & 0 & 1
\end{bmatrix}
\]
which corresponds to a basis \(\{\ket{j}\}_{j\in \mathcal{J}}\), whose elements will be denoted by 

\[
\begin{aligned}
\ket{\phi_{00}} &= \frac{1}{\sqrt{2}}\left(\ket{00} - \mathrm{i}\ket{11}\right) \\
\ket{\phi_{01}} &= \frac{1}{\sqrt{2}}\left(\ket{01} - \mathrm{i}\ket{10}\right) \\
\ket{\phi_{10}} &= \frac{1}{\sqrt{2}}\left(\ket{10} - \mathrm{i}\ket{01}\right) \\
\ket{\phi_{11}} &= \frac{1}{\sqrt{2}}\left(\ket{11} - \mathrm{i}\ket{00}\right)
\end{aligned}
\]

The initial state will be given by 

\[
 \ket{\psi} = \mathcal{J} \ket{00} = \frac{1}{\sqrt{2}} \ket{00} - \mathrm{i} \ket{11}
\]

For each party, consider the following sets of unitary operators $\{U_a, U_a', U_a'', U_a'''\}, \{U_b, U_b', U_b'', U_b'''\}$:\footnote{Any such unitary can be decomposed as
\[
U = Z(\xi_1)\,R_y(\pi/2)\,Z(\xi_2),
\]
where
\[
Z(\xi) = \mathrm{diag}(e^{-i\xi/2}, e^{i\xi/2})
\]
denotes a virtual $Z$ rotation and
\[
R_y(\pi/2) =
\frac{1}{\sqrt{2}}
\begin{bmatrix}
1 & -1\\
1 & 1
\end{bmatrix}
\]
is a physical resonant carrier pulse about the $Y$ axis.

It is worth noting that in settings where local unitaries differ only by $Z$ rotations, implemented either virtually or through passive waveplates, these physical setups may be less susceptible to side-channel attacks than traditional photonic-based communication protocols.}

\[
U_a = \frac{1}{\sqrt{2}}
\begin{bmatrix}
1 & e^{i \pi/4} \\
- e^{-i \pi/4} & 1
\end{bmatrix}, \quad
U_b = \frac{1}{\sqrt{2}}
\begin{bmatrix}
1 & e^{i \pi/4} \\
- e^{-i \pi/4} & 1
\end{bmatrix}
\]

\[
U_a' = \frac{1}{\sqrt{2}}
\begin{bmatrix}
e^{i \pi/2} & e^{i 3\pi/4} \\
- e^{-i 3\pi/4} & e^{-i \pi/2}
\end{bmatrix}, \quad
U_b' = \frac{1}{\sqrt{2}}
\begin{bmatrix}
e^{i \pi/2} & e^{i 3\pi/4} \\
- e^{-i 3\pi/4} & e^{-i \pi/2}
\end{bmatrix}
\]

\[
U_a'' = \frac{1}{\sqrt{2}}
\begin{bmatrix}
e^{i 5\pi/4} & e^{i \pi/2} \\
- e^{-i \pi/2} & e^{-i 5\pi/4}
\end{bmatrix}, \quad
U_b'' = \frac{1}{\sqrt{2}}
\begin{bmatrix}
e^{i 3\pi/4} & 1 \\
-1 & e^{-i 3\pi/4}
\end{bmatrix}
\]

\[
U_a''' = \frac{1}{\sqrt{2}}
\begin{bmatrix}
e^{i 3\pi/4} & 1 \\
-1 & e^{-i 3\pi/4}
\end{bmatrix}, \quad
U_b''' = \frac{1}{\sqrt{2}}
\begin{bmatrix}
e^{i \pi/4} & e^{i 3\pi/2} \\
- e^{-i 3\pi/2} & e^{-i \pi/4}
\end{bmatrix}
\]

After Charley prepares the state $\mathcal{J}\ket{00}$ and distributes the subsystems to the parties, each party applies a local unitary from their respective sets to their subsystem, after which the joint state $\ket{\phi} = (U_A \otimes U_B) \ket{\psi}$ is sent to Eve for measurement.

\begin{figure}[ht]
    \centering
    \[
    \begin{array}{c|cccc}
       & U_b & U_b' & U_b'' & U_b''' \\ \hline
    U_a & \ket{\phi_{00}} & \ket{\phi_{11}} & \ket{\phi_{10}} & \ket{\phi_{01}}\\
    U_a' & \ket{\phi_{11}} & \ket{\phi_{00}} & \ket{\phi_{01}} & \ket{\phi_{10}} \\
    U_a'' & \ket{\phi_{10}} & \ket{\phi_{01}} & \ket{\phi_{00}} & \ket{\phi_{11}} \\
    U_a''' & \ket{\phi_{01}} & \ket{\phi_{10}} & \ket{\phi_{11}} & \ket{\phi_{00}} \\
    \end{array}
    \]
    \caption{Measurement outcomes given corresponding joint unitaries, where rows correspond to Alice's unitaries and columns correspond to Bob's. Each measurement outcome denoted in an entry occurs with probability one under corresponding joint unitaries.}
    \label{fig:game_matrix}
\end{figure}

To perform a measurement in the $\mathcal{J}$ basis, Eve applies the Hermitian conjugate of the basis-defining unitary,
\[
\mathcal{J}^{\dagger} = \frac{1}{\sqrt{2}}
\begin{bmatrix}
1 & 0 & 0 & i \\
0 & 1 & i & 0 \\
0 & i & 1 & 0 \\
i & 0 & 0 & 1
\end{bmatrix}.
\]

This unitary maps each $\mathcal{J}$-basis state back to the computational basis:
\[
\begin{aligned}
\mathcal{J}^{\dagger} \ket{\phi_{00}} &= \ket{00}, &
\mathcal{J}^{\dagger} \ket{\phi_{01}} &= \ket{01}, \\
\mathcal{J}^{\dagger} \ket{\phi_{10}} &= \ket{10}, &
\mathcal{J}^{\dagger} \ket{\phi_{11}} &= \ket{11}.
\end{aligned}
\]

Finally, Eve measures both qubits in the computational basis $\{\ket{00}, \ket{01}, \ket{10}, \ket{11}\}$, and announces the measurement outcome.

The set of feasible joint unitaries yield the distribution over measurement outcomes, shown in Figure \ref{fig:game_matrix}, in an ideal setting with no measurement error. Note that given the knowledge of one's own unitary and correct measurement outcome, each party can determine the other party's chosen unitary with probability one. In this example, as in superdense coding, each party encodes two classical bits through their choice of local unitary. 

\subsection{Mutual Certainty with Locally Information-Free Encodings: Classical Impossibility}
Given the measurement outcome $j$ together with their local inputs $x_i$, each party can infer the others' chosen unitary, and therefore local inputs, with certainty, i.e. $H(X_b | j, x_a) = H(X_a | j, x_b) = 0$.

Note that the parties' local reduced densities $\rho_A, \rho_B$ carry zero Holevo information about $X_a$ and $X_b$ respectively (conditional on all other information in the primitive, including the shared initial state $\ket{\psi}$). This is due to the fact that since $(U_a(x_a) \otimes U_b(x_b))|\psi\rangle\langle\psi|(U_a^\dagger \otimes U_b^\dagger)$ is a pure state on $\mathcal{H}_A \otimes \mathcal{H}_B$, and $|\psi\rangle$ is maximally entangled, the reduced state on Alice's subsystem is 
$$\rho_A(x_a) = \mathrm{Tr}_B[(U_a(x_a) \otimes U_b(x_b))|\psi\rangle\langle\psi|(U_a^\dagger \otimes U_b^\dagger)] = \frac{I}{d}$$
for all $x_a$, and symmetrically $\rho_B(x_b) = \frac{I}{d}$ for all $x_b$. 
Therefore,
\begin{equation}
    I(X_a; \rho_A) = 0, \qquad I(X_b; \rho_B) = 0
\end{equation}
(conditional on the shared initial state $\ket{\psi}$), while the joint density $\rho_{AB}(x_a, x_b)$ permits a measurement event $j$ such that $H(X_b | j, x_a) = H(X_a | j, x_b) = 0$.

\begin{claim}[No classical locally-private, globally-recoverable encoding]\label{proof:classical_impossibility_j} \leavevmode \\
Let $X_a \perp X_b$ with $H(X_a), H(X_b) > 0$, and let $R \perp (X_a, X_b)$ be classical shared randomness.

Consider a classical protocol that produces messages $M_i$ for each party $i$, where $M_i = f_i(X_i, R)$, are functions of party $i$'s local inputs and shared randomness $R$, and a variable $J = g(M_a, M_b)$, which is a function of the parties' messages.

Suppose the protocol is \textit{locally information-free}, such that:
\[
I(X_i; M_i \mid R) = 0 \quad \forall i,
\]
while being \textit{globally recoverable} for the participants, such that for all realizations $x_i$ and $j$:
\[
H(X_{-i} \mid X_i = x_i, J = j) = 0.
\]
No such classical protocol exists.
\end{claim}

In particular, the quantum setting enables purely relational encoding: information about multiple independent inputs exists only in a single joint correlation event, and is absent from all locally accessible subsystems, eliminating the need for a classical shared key that must be stored, protected, and reused.

\subsection{Determining $s(x_a, x_b)$}
Recall that following Eve's measurement announcement $j$, each party records a value $s_i$, which depends on both their local information and Eve's announcement.
An example of a recording strategy that satisfies the properties described in section \ref{sec:properties:recording-s} is depicted in Figure \ref{fig:xa+2xb mod 4}. 
Note that under each joint unitary, the parties obtain identical values $s_a = s_b$. For convenience, we denote this shared value under the true joint unitaries as $\tilde{s} \in \mathcal{A} \cup \emptyset$.

    \begin{figure}[ht]
    \centering
    \[
    \begin{array}{c|cccc}
       & U_b & U_b' & U_b'' & U_b''' \\ \hline
     U_a & (\ket{\phi_{00}}, {3}) & (\ket{\phi_{11}}, {1}) & (\ket{\phi_{10}}, {3}) & (\ket{\phi_{01}}, {1})\\
    U_a' & (\ket{\phi_{11}}, {0}) & (\ket{\phi_{00}}, {2}) & (\ket{\phi_{01}}, {0}) & (\ket{\phi_{10}}, {2}) \\
    U_a'' & (\ket{\phi_{10}}, 1) & (\ket{\phi_{01}}, 3) & (\ket{\phi_{00}}, 1) & (\ket{\phi_{11}}, 3) \\
    U_a''' & (\ket{\phi_{01}}, 2) & (\ket{\phi_{10}}, 0) & (\ket{\phi_{11}}, 2) & (\ket{\phi_{00}}, 0) \\
    \end{array}
    \]
    \caption{Each matrix element corresponds with pairs $(j, s_i)$, indicating parties' recording strategies $s_i$ under their corresponding input $U_i$, and the measurement outcome $j$. Note that under each joint unitary, the parties obtain identical values $s_a = s_b = \tilde{s}$.}
    \label{fig:xa+2xb mod 4}
\end{figure}

In particular, for any such valid recording strategy $j \perp \tilde{s}$, and therefore, any external party who may learn $j$, cannot determine $\tilde{s}$ better than uniform random.

Let $\hat{s} = (\hat{s}_a, \hat{s}_b)$ denote an adversary’s estimate of the parties' recorded values $s = (s_a, s_b)$.

Given information $\tilde{Z}$, the adversary forms a posterior $P(U^{(k)} | \tilde{Z})$ over the possible joint unitaries $U^{(k)} \in \mathcal{U}_{a,b}$. The probability of a successful estimate, conditioned on the information $\tilde{Z}$ and the announcement $j$, is:
\begin{equation}
    P(\hat{s} = s | \tilde{Z}, j) = \sum_{k=1}^{|\mathcal{U}_{a, b}|} P(\hat{s} = s | U^{(k)}, j) P(U^{(k)} | \tilde{Z})
\end{equation}

where:
\begin{itemize}
    \item $P(U^{(k)} \mid \tilde{Z})$ is the adversary's posterior distribution over the feasible joint unitary $U^{(k)}$ given the observed information $\tilde{Z}$. 
    
    \item $P(\hat{s} = s \mid U^{(k)}, j)$ is a deterministic value $\{0, 1\}$ defined by the recording strategies. Specifically, for a fixed $U^{(k)}$ and $j$, the parties' recording strategy uniquely determines the pair $(s_a, s_b)$. The adversary's estimate $\hat{s}$ is the value $(s_a, s_b)$ under the unitary $U^{(k)}$ and the announcement $j$ given by the recording strategies.
\end{itemize}

\begin{claim}\label{claim:prob_bound_adversary}
    For any adversary who obtains information $\tilde{Z}$,
    if $s_a = s_b = \tilde{s}, \tilde{s} \neq \emptyset$, then $P(\hat{s} = s | \tilde{Z}, j) = P(U | \tilde{Z})$ for the joint unitary $U$ chosen by the parties.
\end{claim}

Therefore, an adversarial Charley or Eve's probability of correctly estimating the parties' recorded values, $P(\hat{s} = s | \tilde{Z}, j) = P(U | \tilde{Z})$, the probability of determining the true joint unitary given any information $\tilde{Z}$ obtained from the protocol. This is the problem of quantum state discrimination.

\subsection{Adversarial Success Bounds}\label{sec:adversarial_success_bounds}
Suppose that \textit{prior} to the application of the local unitaries, Charley appends an ancilla system in $\mathcal{H}_E$ with dimension $d_E$. The resulting joint initial state is:
\[
|\psi\rangle_{ABE} \in \mathcal{H}_a \otimes \mathcal{H}_b \otimes \mathcal{H}_E,
\]
where $|\psi\rangle_{AB}$ and $|\psi\rangle_E$ denote the reduced states on the $AB$ system and the ancilla system, respectively. 

The parties then apply their respective local unitaries, yielding the state:
\[
|\phi\rangle_{ABE} = \left( U_a \otimes U_b \otimes I_E \right) |\psi\rangle_{ABE}.
\]

Following the local operations, we define the post-unitary density operators as: \\
\[
\rho_{ABE} = |\phi\rangle\langle\phi|_{ABE}, \quad \rho_{AB} = \text{Tr}_E[\rho_{ABE}], \quad \rho_{E} = \text{Tr}_{AB}[\rho_{ABE}]
\]

\begin{definition}
     The \textit{Protocol Execution } $\mathcal{T}$ represents the vector of all elements produced by the protocol, excluding parties' local private inputs $\{x_i\}$ and local randomness $\{r_i\}$.

\begin{itemize}
    \item Classical Case: $\mathcal{T}_C = (T, R)$, where $R$ is the initial shared randomness and $T = \{m_1, m_2, \dots, m_l\}$ is the the transcript, or the set of all messages exchanged. Each message $m_k$ is a function $f_k(x_k, R, \mathcal{M}_{<k})$ of a party's local input $x_k$, the shared randomness $R$, and the set of all preceding messages $\mathcal{M}_{<k}$.
    \item Quantum Case: $\mathcal{T}_Q = (\rho_{ABE}, \ket{\psi}_{ABE}, j)$, where $\rho_{ABE}$ is the post-unitary density matrix, $\ket{\psi}_{ABE}$ is the initial shared state, and $j$ is the public announcement.
\end{itemize}

\end{definition}

\begin{claim}\label{claim:adversary_bounds_entangled}
For any $\mathcal{T}_Q = (\rho_{ABE}, \ket{\psi}_{ABE}, j)$, let $\tilde{Z}$ denote the information obtained from any quantum operation or measurements performed on the post-unitary system $\rho_{ABE}$, where the operation may conditioned on knowledge of the initial state $\ket{\psi}_{ABE}$ and any classical communication between adversarial parties.

In the described primitive, if the initial reduced state $\ket{\psi}_{AB}$ is maximally entangled: \\
$P(\hat{s} = s \mid \tilde{Z}, j) \leq \frac{1}{4} = \frac{1}{|\mathcal{A}|}$, given the extractable information $\tilde{Z}$ and the public announcement $j \in \mathcal{T}_Q$.

\end{claim}

\begin{claim}\label{claim:adversary_bounds_locc}
For any $\mathcal{T}_Q = (\rho_{ABE}, \ket{\psi}_{ABE}, j)$, let $\tilde{Z}$ denote the information obtained via any \textit{Local Operations and Classical Communication} (LOCC) protocol performed on the post-unitary systems $\rho_{AB}$ and $\rho_E$. The choice of local measurements and the subsequent classical exchange may be conditioned on knowledge of the initial state $\ket{\psi}_{ABE}$.

In the described primitive, under any feasible initial reduced state $\ket{\psi}_{AB}$ (which need not be entangled):\\
$P(\hat{s} = s \mid \tilde{Z}, j) \leq \frac{1}{4} = \frac{1}{|\mathcal{A}|}$,  given the extractable information $\tilde{Z}$ and the public announcement $j \in \mathcal{T}_Q$. 
\end{claim}

In the described primitive, when $s_a = s_b = \tilde{s} \neq \emptyset$, under the adversarial attacks described Claim \ref{claim:adversary_bounds_entangled} and Claim \ref{claim:adversary_bounds_locc}, $P(\hat{s} = s \mid \tilde{Z}, j) \leq \frac{1}{4} = \frac{1}{|\mathcal{A}|}$. In these settings, Charley and Eve either learn $s$ with probability no better than random, or $s_a \neq s_b$. These bounds require that either $\ket{\psi}_{AB}$ is maximally entangled, \textit{or} Charley and Eve are physically separated so to ensure that each may only measure resulting systems $\rho_E, \rho_{AB}$, respectively and prohibits the malicious parties from quantum communicating.\footnote{Note that trusted state preparation is a requirement of the standard MDI-QKD protocol, and security guarantees require that third-party measurement devices interact only with the signals sent through the parties' quantum channels. Similarly, in superdense coding, if an adversary holds an ancilla entangled with the initial state and can measure it jointly with the post-unitary state, they can learn the encoded message.}. For a discussion on decoupling success rates from error rates see section \ref{appx:decoupling_succ_error}.

\begin{claim}[Classical Impossibility]\label{claim:classical_impossiblity_s}
In any classical protocol with execution $\mathcal{T}_C = (T, R)$, where $s_a = s_b = \tilde{s} \neq \emptyset$, 
\begin{equation}
    P(\hat{s} = s \mid \mathcal{T}_C) > \frac{1}{|\mathcal{A}|}.
\end{equation}
Conversely, by Claim \ref{claim:adversary_bounds_entangled}, the quantum execution $\mathcal{T}_Q = (\rho_{ABE}, \ket{\psi}_{ABE}, j)$ permits honest parties to generate the same value $s_a = s_b = \tilde{s}$ while ensuring $P(\hat{s} = s \mid \tilde{Z}, j) \leq \frac{1}{|\mathcal{A}|}$, when the initial resource $\ket{\psi}_{AB}$ is \textit{maximally-entangled}.
\end{claim}

The proposed setting therefore permits the parties to generate $\tilde{s}$ that remains information-theoretically hidden, even from adversaries who may condition all on extractable information generated by the protocol-- bypassing the classical limitation that requires a pre-shared classical key which must be stored, reused, and inaccessible to an adversary, in order to securely generate $\tilde{s}$. 

A detailed discussion on various implementations of the proposed primitive, including QKD protocols, and shared function-generation where values depend non-degenerately on both parties' local inputs, as well a discussion on generating private, unbiased coins are relegated to additional publications.

\newpage

\printbibliography

@article{Shi2016,
  title={Secure Multiparty Quantum Computation for Summation and Multiplication},
  author={Shi, Run-hua and others},
  journal={Scientific Reports},
  volume={6},
  pages={19655},
  year={2016},
  doi={10.1038/srep19655}
}

@article{pan2024evolution,
  title={The evolution of quantum secure direct communication: On the road to the qinternet},
  author={Pan, Dong and Long, Gui-Lu and Yin, Liuguo and Sheng, Yu-Bo and Ruan, Dong and Ng, Soon Xin and Lu, Jianhua and Hanzo, Lajos},
  journal={IEEE Communications Surveys \& Tutorials},
  volume={26},
  number={3},
  pages={1898--1949},
  year={2024},
  publisher={IEEE}
}

@article{yang2020bidirectional,
  title={Bidirectional and cyclic quantum dense coding in a high-dimension system: X. Yang et al.},
  author={Yang, Xue and Bai, Ming-qiang and Mo, Zhi-wen and Xiang, Yi},
  journal={Quantum Information Processing},
  volume={19},
  number={2},
  pages={43},
  year={2020},
  publisher={Springer}
}

@article{nguyen2004quantum,
  title={Quantum dialogue},
  author={Nguyen, Ba An},
  journal={Physics Letters A},
  volume={328},
  number={1},
  pages={6--10},
  year={2004},
  publisher={Elsevier}
}

@article{maitra2017measurement,
  title={Measurement device-independent quantum dialogue},
  author={Maitra, Arpita},
  journal={Quantum Information Processing},
  volume={16},
  number={12},
  pages={305},
  year={2017},
  publisher={Springer}
}

@article{Lu2024,
  title={Quantum Secure Multi-Party Summation with Graph State},
  author={Lu, Yan and Ding, Guihua},
  journal={Entropy},
  volume={26},
  number={1},
  pages={80},
  year={2024},
  publisher={MDPI},
  doi={10.3390/e26010080}
}

@article{wu2023quantum,
  title={Quantum secure multi-party summation using single photons},
  author={Wu, Wan-Qing and Xie, Ming-Zhe},
  journal={Entropy},
  volume={25},
  number={4},
  pages={590},
  year={2023},
  publisher={MDPI}
}

@article{yi2021quantum,
  title={Quantum secure multi-party summation protocol based on blind matrix and quantum Fourier transform},
  author={Yi, Xin and Cao, Cong and Fan, Ling and Zhang, Ru},
  journal={Quantum information processing},
  volume={20},
  number={7},
  pages={249},
  year={2021},
  publisher={Springer}
}

@article{hausladen1994pretty,
  title={A ‘pretty good’measurement for distinguishing quantum states},
  author={Hausladen, Paul and Wootters, William K},
  journal={Journal of Modern Optics},
  volume={41},
  number={12},
  pages={2385--2390},
  year={1994},
  publisher={Taylor \& Francis}
}

@article{eldar2004optimal,
  title={Optimal detection of symmetric mixed quantum states},
  author={Eldar, Yonina C and Megretski, Alexandre and Verghese, George C},
  journal={IEEE Transactions on Information Theory},
  volume={50},
  number={6},
  pages={1198--1207},
  year={2004},
  publisher={IEEE}
}

@inproceedings{Yao1982,
  author    = {Andrew C. Yao},
  title     = {Protocols for Secure Computations},
  booktitle = {23rd Annual Symposium on Foundations of Computer Science (FOCS)},
  year      = {1982},
  pages     = {160--164}
}

@inproceedings{GoldreichMicaliWigderson1987,
  author    = {Oded Goldreich and Silvio Micali and Avi Wigderson},
  title     = {How to Play any Mental Game — A Completeness Theorem for Protocols with Honest Majority},
  booktitle = {19th Annual ACM Symposium on Theory of Computing (STOC)},
  year      = {1987},
  pages     = {218--229}
}

@article{CrepeauGottesmanSmith2002,
  author    = {Claude Cr\'epeau and Daniel Gottesman and Adam D. Smith},
  title     = {Secure Multi‑party Quantum Computing},
  journal   = {arXiv e‑prints},
  eprint    = {quant‑ph/0206138},
  year      = {2002},
  url       = {https://arxiv.org/abs/quant‑ph/0206138}
}

@article{bennett1992communication,
  title={Communication via one-and two-particle operators on Einstein-Podolsky-Rosen states},
  author={Bennett, Charles H and Wiesner, Stephen J},
  journal={Physical review letters},
  volume={69},
  number={20},
  pages={2881},
  year={1992},
  publisher={APS}
}

@article{nathanson2005distinguishing,
  title={Distinguishing bipartitite orthogonal states using LOCC: Best and worst cases},
  author={Nathanson, Michael},
  journal={Journal of Mathematical Physics},
  volume={46},
  number={6},
  year={2005},
  publisher={AIP Publishing}
}

\newpage
\section{Appendix}\label{ref:appendix}

\subsection{Proof of Claim \ref{proof:classical_impossibility_j}}
\begin{proof}
First, since $I(X_i; M_i | R) = 0$ $\forall i$, $P(m_i | x_i, r) = P(m_i | r)$, for all realized messages, local inputs, and shared randomness $m_i, x_i, r$. 

Second, note that for all messages $m_a, m_b$ sent by Alice and Bob respectively,
$$P(m_a, m_b | x_a, x_b) = \sum_{r \in \mathcal{R}} P(r) P(m_a| x_a, r) P(m_b | x_b, r)$$
since $R \perp (X_a, X_b)$ and the messages are independent given the shared randomness $M_a \perp M_{b} | (X_a, X_b, R)$. 

Additionally, we may write
$$P(j | x_a, x_b) = \sum_{m_a, m_b} \sum_{r \in \mathcal{R}} P(r) P(j | m_a, m_b)P(m_a| x_a, r) P(m_b | x_b, r)$$
and by the locally information-free condition

$$P(j | x_a, x_b) = \sum_{m_a, m_b} \sum_{r \in \mathcal{R}} P(r) P(j | m_a, m_b)P(m_a| r) P(m_b | r) = P(j)$$

Thus $J \perp (X_a, X_b)$ and the announcement $j$, does not depend on inputs $(x_a, x_b)$, and therefore $H(X_b | x_a, j) >0, H(X_a | x_b, j) > 0$.
\end{proof}

\subsection{Properties of the Recording Strategy}\label{sec:properties:recording-s}

For every joint unitary denote
\[
U = \bigotimes_{i=a}^{b} U_i^{(k_i)} \in \bigotimes_{i=a}^{b} \mathcal{U}_i,
\]

and for any projective measurement described by the set of orthogonal projectors 
\(\{\Pi_j = \ket{j}\!\bra{j}\}_{j \in \mathcal{J}}\)
on the post-unitary state
\(\ket{\phi} = U \ket{\psi}\),
the conditional probability of obtaining outcome \(j\) is
\[
\Pr(j \mid U, \ket{\psi}) = 
\bra{\phi} \Pi_j \ket{\phi} = 
\bra{\psi} U^\dagger \Pi_j U \ket{\psi}.
\]
We define
\[
j^*(U) = \arg\max_{j \in \mathcal{J}} \Pr(j \mid U, \ket{\psi}),
\]
as the most likely measurement outcome in basis \(J\) when measuring the post-unitary state \(\ket{\phi}\) (which may not necessarily be unique).

Define a correspondence
\[
f : \bigotimes_{i=a}^{b} \mathcal{U}_i \longrightarrow J \times \left(\mathcal{A} \cup \{\emptyset\}\right),
\qquad 
f(U) = \big(j^*(U), \tilde{s}\big),
\]
associating each joint unitary with the most likely measured outcome  and the parties’ recorded values.

For each fixed index \(i\), denote 
\[
\mathcal{U}_{-i}=\prod_{k\neq i}\mathcal{U}_k
\]
the Cartesian product of the feasible sets of all parties' unitaries except \(i\). A joint unitary may be written as
\[
U=(U_i,U_{-i})\in\mathcal{U}_i\times\mathcal{U}_{-i}.
\]

While the mapping $f(U)$ may in general be multi-valued, one may select a single feasible function from the correspondence. All such selections require that, the measurement basis $\mathcal{J}$, the initial state $\ket{\psi}$, the feasible local sets 
$\{\mathcal{U}_i\}_{i=a}^{b}$, and the recorded values $\tilde{s}$ satisfy either the first, or both, of the following properties.

\paragraph{Recoverability.}
Let $f^j$ denote the restriction of $f$ to its $j^*$ component:
\[
f^j : \bigotimes_{i=a}^{b} \mathcal{U}_i \to J, \quad 
f^j(U) = j^*(U).
\]

We require that \(f^j\) be injective in each argument when the other arguments are held fixed. Equivalently, for every \(i\in\{a, b\}\) the following condition holds:
\[
\forall\,U_{-i}\in\mathcal{U}_{-i}:\qquad
U_i\mapsto f^j(U_i,U_{-i})\ \text{is injective on }\mathcal{U}_i,
\]
and
\[
\forall\,U_{i}\in\mathcal{U}_{i}:\qquad
U_{-i}\mapsto f^j(U_i,U_{-i})\ \text{is injective on }\mathcal{U}_{-i},
\]

 In other words, for any fixed choices of the other parties, distinct local unitaries of party \(i\) produce distinct most-likely outcomes, and for any fixed own-unitary, distinct unitaries of other parties produce distinct most-likely outcomes. Therefore, given their own unitary \(U_i\) and the announced outcome \(j^*=f^j(U_i,U_{-i})\), party \(i\) can uniquely recover the tuple \(U_{-i}\) (and hence each other party's unitary), when $j^*$ is measured under $(U_i, U_{-i})$ in $\mathcal{J}$ with probability one.

\paragraph{Indistinguishibility.}
  \begin{enumerate}
      \item For every $U$ with $f(U) = (j^*, \tilde{s})$ with $\tilde{s} \neq \emptyset$, the mapping $f$ is injective: 
      
      \[
          f(U) = f(U') \implies U = U'.
      \]
    Each $(j^*, \tilde{s})$ with a non-null $\tilde{s}$ value, corresponds to a unique joint unitary.

      \item  
      For every $j^*$ for which there is a $U$ where $f(U) = (j^*, \tilde{s})$ and $\tilde{s} \neq \emptyset$, 
         there exists a set of joint unitaries 
    \[
    \mathcal{U}_{j^*} \subseteq \bigotimes_{i=a}^{b} \mathcal{U}_i
    \]
    such that $f^j(U_j) = j^*$ and $Pr(j^* | U_j, \ket{\psi}) \approx Pr(j^* | U, \ket{\psi})$ for all $U_j \in \mathcal{U}_{j^*}$ and $|\mathcal{U}_{j^*}| \geq |\mathcal{A}|$.

    Hence, multiple distinct joint unitaries correspond to the same most-likely measurement outcome \( j^* \) with approximately equal likelihood, ensuring that knowledge of \( j^* \) alone does not allow any observer to infer the recorded value 
\( \tilde{s} \) with probability greater than $\frac{1}{q}$.

  \end{enumerate}

When the initial state, basis, and unitaries satisfy the recoverability property only, the parties may deduce each other's local inputs from the measurement outcomes, although the recording strategy may not be such that the recorded values of $\tilde{s}$ remain fully hidden from all observers or attackers. Any such primitive facilitates the decentralized exchange of local private information.
In these cases, the particular security bounds on joint unitaries and recorded values of $\tilde{s}$ will depend on the number of unitaries in the feasible set and the recording strategy chosen.

An example of such an $f$ is depicted in Figure \ref{fig: recoverability_only}, for two parties $i \in \{a, b\}$, for $j^* \in \{j_1, j_2, j_3, j_4\}$, and a binary function value $\tilde{s} \in \{s_1, s_2\}$. 
\begin{figure}[ht]
    \centering
    \[
    \begin{array}{c|cccc}
       & U_b & U_b' & U_b'' & U_b''' \\ \hline
    U_a & (j_1, s_1) & (j_2, s_1) & (j_3, s_1) & (j_4, s_1)\\
    U_a' & (j_2, s_1) & (j_1, s_1) & (j_4, s_1) & (j_3, s_2) \\
    U_a'' & (j_3, s_1) & (j_4, s_1) & (j_1, s_2) & (j_2, s_2) \\
    U_a''' & (j_4, s_1) & (j_3, s_2) & (j_2, s_2) & (j_1, s_2) \\
    \end{array}
    \]
    \caption{Entries given by values of $f(U) = (j^*, \tilde{s})$ for joint unitaries $U$, and corresponding recording strategies $\tilde{s}$ given unitaries $U_i$ and measurement announcements $j^*$.}
    \label{fig: recoverability_only}
\end{figure}

Instead, an example of a recording strategy that that satisfies \textit{both} properties for a valid $f$ is depicted in Figure \ref{fig: arb_s_strategy}, for two parties $i \in \{a, b\}$, for $j^* \in \{j_1, j_2, j_3, j_4\}$, and $\tilde{s} \in \{s_1, s_2, s_3, \emptyset\}$. We note that for any given $\mathcal{U}, \mathcal{J}, \mathcal{A}$, there may be many valid recording strategies that satisfy the prescribed properties.

\begin{figure}[ht]
    \centering
    \[
    \begin{array}{c|cccc}
       & U_b & U_b' & U_b'' & U_b''' \\ \hline
    U_a & (j_1, s_1) & (j_2, s_2) & (j_3, s_3) & (j_4, \emptyset)\\
    U_a' & (j_2, s_1) & (j_1, s_2) & (j_4, s_3) & (j_3, \emptyset) \\
    U_a'' & (j_3, s_2) & (j_4, s_1) & (j_1, \emptyset) & (j_2, s_3) \\
    U_a''' & (j_4, s_2) & (j_3, s_1) & (j_2, \emptyset) & (j_1, s_3) \\
    \end{array}
    \]
    \caption{Entries given by values of $f(U) = (j^*, \tilde{s})$ for joint unitaries $U$, and corresponding recording strategies $\tilde{s}$ given unitaries $U_i$ and measurement announcements $j^*$.}
    \label{fig: arb_s_strategy}
\end{figure}

\section{Proof of Claim \ref{claim:prob_bound_adversary}}
\begin{proof}
    If $s_a = s_b = \tilde{s}$, then $j = j^*(U)$ under the true joint unitary $U$. If the adversary announces $j \neq j^*(U)$, parties will record a values of $\tilde{s}_i$ that are either non-matching or null. This follows from the recoverability property, and indistinguishiblity property one -- each $(j^*, \tilde{s})$ pair for any $\tilde{s} \neq \emptyset$ is unique.

    Next, consider any $\tilde{U} \neq U$. If $j^*(\tilde{U}) = j^*(U)$, then $\hat{s}_a \neq s_a, \hat{s}_b \neq s_b$, and $P(\hat{s} = s | \tilde{U}, j) = 0$. This follows from indistinguishibility property two.

    If instead, $j^*(\tilde{U}) \neq j^*(U)$, then $\hat{s}_a \neq \hat{s_b}$ and $P(\hat{s} = s | \tilde{U}, j) = 0$ by the recoverability property, and indistinguishiblity property one -- each $(j^*, \tilde{s})$ pair for any $\tilde{s} \neq \emptyset$ is unique.

    Therefore, $P(\hat{s} = s | \tilde{Z}, j) = \sum_{i=k}^{|\mathcal{U}_{a, b}|} P(\hat{s} = s | U^{(k)}, j) P(U^{(k)} | \tilde{Z}) = P(\hat{s} = s | U, j)P(U |\tilde{Z})$. 
    Since $P(\hat{s} = s | U, j) = 1$ under $(U, j)$, $P(\hat{s} = s | \tilde{Z}, j) = P(U | \tilde{Z})$.
    
\end{proof}

\section{Proof of Claim \ref{claim:adversary_bounds_entangled}}
We prove the claim by showing that any arbitrary quantum operation reduces to a state discrimination problem on the $\rho_{AB}$ systems, bounded by the properties of the Pretty Good Measurement (PGM).


\begin{lemma}
    Any POVM performed jointly on $\rho_{ABE}$ provides no more information than a measurement on $\rho_{AB}$ alone when the global state factorizes as $\rho_{AB} \otimes \rho_E$.
\end{lemma}

\begin{proof}
    Suppose that Eve can measure the post-unitary state jointly with any initial ancilla system, 
\[
\ket{\phi}
    = \big(  \bigotimes_{i=a}^{b} U_i \otimes I_E \big)\ket{\psi_{ABE}}.
\]

If the initial state shared between Alice and Bob is maximally entangled,
then the joint state must factorize as
\[
\ket{\psi}_{ABE} = \ket{\psi}_{AB} \otimes \ket{\psi}_E,
\]
since a maximally entangled bipartite state cannot be correlated with a third
system. Therefore the resulting density matrix after the unitaries are applied
is of the form
\[
\rho_{ABE} = \rho_{AB} \otimes \rho_E.
\]

Consider any POVM $\{M_j\}$ acting on the joint system $ABE$. The probability of
measurement outcome $j$ is
\[
\mathrm{Tr}_{ABE}[M_j (\rho_{AB} \otimes \rho_E)].
\]
Using properties of the partial trace, this can be written as
\[
\mathrm{Tr}_{ABE}[M_j (\rho_{AB} \otimes \rho_E)]
=
\mathrm{Tr}_{AB}\!\left[
    \mathrm{Tr}_E[M_j (I_{AB} \otimes \rho_E)]\, \rho_{AB}
\right].
\]

Define
\[
N_j = \mathrm{Tr}_E[M_j (I_{AB} \otimes \rho_E)].
\]
Then
\[
\mathrm{Tr}_{ABE}[M_j (\rho_{AB} \otimes \rho_E)]
=
\mathrm{Tr}_{AB}[N_j \rho_{AB}].
\]

It remains to show that $\{N_j\}$ forms a valid POVM on the $AB$ system.
Positivity follows because the partial trace preserves positivity.
Furthermore,
\[
\sum_j N_j
=
\mathrm{Tr}_E\!\left[\left(\sum_j M_j\right)(I_{AB} \otimes \rho_E)\right]
=
\mathrm{Tr}_E[I_{ABE} (I_{AB} \otimes \rho_E)]
=
I_{AB} \mathrm{Tr}(\rho_E)
=
I_{AB}.
\]
Thus $\{N_j\}$ is a valid POVM on the $AB$ subsystem. Therefore, any measurement performed jointly on $ABE$ provides no more information than a measurement on $AB$ alone when the global state factorizes as $\rho_{AB} \otimes \rho_E$.
\end{proof}

The PGM protocol (\cite{hausladen1994pretty}) solves for POVMs $\{\Pi_j\}$ given density matrices $\{\rho^j\}$ each occurring with prior probability $p_j$ such that:

$$\Gamma = \sum_j p_j \rho^j$$
$$\pi_j = p_j \Gamma^{-1/2} \rho^j \Gamma^{-1/2} $$

$$P_{succ}^{PGM,j} = Tr[\rho^j \Pi_j]$$
is the probability of detecting $\rho^j$ under PGM, when it is the true state.

By Corollary 2 in \cite{eldar2004optimal}, PGM is optimal if the probability of correctly detecting the state is identical for any state $\rho^j$. 
In other words if 
$$P_{succ}^{PGM,j} = Tr[\rho^j \Pi_j] = P_{succ}^{PGM,k} = Tr[\rho^k \Pi_k], \forall j, k$$ then PGM achieves the maximal state distinguishibility.

Moreover, the following condition holds: 

\begin{lemma}
    $\frac{1}{M} \sum_j Tr[\pi_j \rho^j] \leq \frac{d}{M}$, for $M$ states $\rho^j$, each in $d$ dimensions.
\end{lemma}

\begin{proof}
Since $\rho^j$ is a positive semi-definite pure-state density matrix, its eigenvalues are $\{1,0\}$ and therefore
\[
\lambda_{\max}(\rho^j)=1.
\]
The eigenvalues of $I-\rho^j$ are then given by $1-\lambda_k(\rho^j)$, where $\lambda_k(\rho^j)$ denote the eigenvalues of $\rho^j$. Since $0 \le \lambda_k(\rho^j) \le 1$, we have
\[
1-\lambda_k(\rho^j) \ge 0.
\]
Hence all eigenvalues of $I-\rho^j$ are non-negative, implying
\[
I-\rho^j \succeq 0,
\]
and in particular
\[
\lambda_{\max}(I-\rho^j) \ge 0.
\]

Since $\pi_j$ is positive semi-definite
$$Tr[\pi_j \lambda_{max} (I - \rho^j)] \geq 0$$

Re-arranging, 

$$ Tr[\pi_j \rho^j] \leq \lambda_{max} Tr[\pi_j] $$
$$\sum_j Tr[\pi_j \rho^j] \leq \sum_j Tr[\pi_j] = d$$
and therefore, $$\frac{1}{M} \sum_j Tr[\pi_j \rho^j] \leq \frac{d}{M} $$

\end{proof}

Therefore, for $M$ quantum states in $d$ dimensions, if $P_{succ}^j = P_{succ}^k, \forall j, k$  $P_{succ}^j \leq \frac{d}{M}, \forall j$. 

The resulting states in the described primitive satisfy the above conditions. Additionally, their ensemble is maximally-mixed, $\frac{1}{M} \sum_{i=1}^{M} \rho^i = \frac{1}{d} I$.

\begin{lemma}
    For a quantum state ensemble $\{ \rho^i \}_{i=1}^{M}$ whose average is maximally mixed, i.e., $\frac{1}{M} \sum_{i=1}^{M} \rho^i = \frac{1}{d} I$, any trace-non-preserving (TNP) operation $E$ followed by a post-selected guess cannot identify the true state $\rho^i$ with a posterior probability higher than the success probability achieved by the Pretty Good Measurement (PGM).
\end{lemma}

\begin{proof}
A TNP operation is defined by a Kraus operator $E$ such that $E^\dagger E = S \leq I$. Let the event $\mathcal{Z}$ denote the "success" of this operation. The probability of this event occurring, given the true state is $\rho^i$, is:
\begin{equation}
    P(\mathcal{Z} \mid \rho^i) = \mathrm{Tr}(E \rho_i E^\dagger) = \mathrm{Tr}(S \rho_i).
\end{equation}    

The probability that the true state is $\rho^i$, conditioned on $\mathcal{Z}$, is:
\begin{equation}
    P(\rho^i \mid \mathcal{Z}) = \frac{P(\mathcal{Z} \mid \rho^i) P(\rho^i)}{P(\mathcal{Z})} = \frac{\frac{1}{M} \mathrm{Tr}(S \rho_i)}{\frac{1}{M} \sum_{j=1}^{M} \mathrm{Tr}(S \rho_j)}.
\end{equation}

By the linearity of the trace and the given condition $\frac{1}{M} \sum_j \rho^j = \frac{1}{d} I$, the denominator (the total probability of the outcome $\mathcal{Z}$) simplifies:
\begin{equation}
    P(\mathcal{Z}) = \mathrm{Tr}\left( S \left[ \frac{1}{M} \sum_{j=1}^{M} \rho^j \right] \right) = \mathrm{Tr}\left( S \frac{1}{d} I \right) = \frac{1}{d} \mathrm{Tr}(S).
\end{equation}
Substituting this back into the posterior probability:
\begin{equation}
    P(i \mid \mathcal{Z}) = \frac{\mathrm{Tr}(S \rho^i)}{\mathrm{Tr}(S) \cdot \frac{M}{d}}.
\end{equation}

Since $S, \rho_i$ are positive semi-definite operators,
 Then:
\begin{equation}
    \mathrm{Tr}(S \rho^i) \leq \lambda_{max}(\rho^i) \mathrm{Tr}(S).
\end{equation}
where $\lambda_{max}(\rho^i)$ is the maximum eigenvalue of $\rho_i$.
Thus, for any $S$, the posterior probability is bounded by:
\begin{equation}
    P(i \mid \mathcal{Z}) \leq \frac{\lambda_{max}(\rho^i) \mathrm{Tr}(S)}{\mathrm{Tr}(S) \cdot \frac{M}{d}} = \frac{d}{M} \lambda_{max}(\rho^i).
\end{equation}
$$P_{succ}^{PGM,i} = Tr[\rho^i \Pi_i] = \frac{d}{M} \lambda_{max}(\rho^i) = \frac{d}{M}$$

Therefore, $P(\rho^i \mid \mathcal{Z}) \leq P_{succ}^{PGM,i}$.

\end{proof}
Therefore, given any quantum operation—global or local, trace-preserving or probabilistic—the success probability in identifying the true state $\rho^i$ (and therefore the joint unitary $U$) in any iteration is bounded by $\frac{1}{4}$.

\section{Proof of Claim \ref{claim:adversary_bounds_locc}}

Let the true joint state be
\(\rho_{ABE}^j \in \mathcal{L}(\otimes_{i=a}^{b} \mathcal{H}_i \otimes\mathcal{H}_E)\), where  the superscript 
$j$ indexes the possible ensemble member, while the subscripts denote the corresponding subsystems. Write \(d_{AB} \coloneqq \dim(\otimes_{i=a}^{b} \mathcal{H}_i\)) and \(d_E \coloneqq \dim(\mathcal{H}_E)\).

Recall that prior to the application of local unitaries, Charley prepares an ancilla system 
\(\mathcal{H}_E\) of arbitrary dimension \(d_E\).  
The resulting joint initial state is
\[
\ket{\psi_{ABE}^{\rm init}} \in \otimes_{i=a}^{b} \mathcal{H}_i \otimes \mathcal{H}_E
\]

The parties then apply their respective local unitaries \(U_{AB}(j) = U_A(j)\otimes U_B(j)\) on the AB system, yielding
\[
\ket{\psi_{ABE}^{j}} = (U_A(j) \otimes U_B(j) \otimes I_E)\ket{\psi_{ABE}^{\rm init}}.
\]

Lemma 8 of \cite{nathanson2005distinguishing} states that for \(M\) equally probable states \(\{\ket{\psi_{ABE}^{j}}\}\) in dimension \(d_{AB} \times d_E\) such that 
\(\ket{\psi_{ABE}^{j}} = (I_E \otimes \tilde{U}(j)) \ket{\psi_{ABE}^{1}}\), for some $ \tilde{U}$, for all $j$ the maximum success probability using LOCC is bounded by
\[
P_{\rm succ}^{\rm LOCC} = \sum_j p_j P_{\rm succ}^{\rm LOCC, j}\le \frac{d_{AB}}{M}.
\]

In other words, when Eve can unilaterly transform any $\ket{\psi_{ABE}^i}$ into any $\ket{\psi_{ABE}^j}$ for any $i, j$ the parties' joint distinguishibility is bounded by $\frac{d_{AB}}{M}$.

\begin{lemma}
    In the described primitive, $P_{\rm succ}^{\rm LOCC, j} \leq \frac{d_{AB}}{M}$
\end{lemma}

\begin{proof}
For some fixed $U_A(1) \otimes U_B(1)$ in the feasible set of joint unitaries, let
$$\ket{\psi_{ABE}^{(1)}} = (I_E \otimes U_A(1) \otimes U_B(1)) \ket{\psi_{ABE}^{\rm init}}$$

Each
\[
W_j := U_{AB}(j) U_{AB}(1)^\dagger = \big(U_A(j) U_A(1)^\dagger\big) \otimes \big(U_B(j) U_B(1)^\dagger\big)
\]
is a unitary on AB 
where \(\ket{\psi_{ABE}^{(j)}} = (I_E \otimes W_j) \ket{\psi_{ABE}^{(1)}}\)

Finally, we note that if there exists a LOCC protocol achieving $P_{\rm succ}^{\rm i, LOCC} > \frac{4}{M}$ for some state $i$, 
then for any other state $j$, Eve can prepend a single local unitary rotation $W_{j,i}$ on the AB system 
that maps $|\psi^{(j)}\rangle \mapsto |\psi^{(i)}\rangle$ before executing the same LOCC protocol for $i$, where the resulting success probability is identical to that for $i$. 
Consequently, the maximal success probability is the same for all states, and by Lemma~8 of~\cite{nathanson2005distinguishing}, we have
\[
P_{\rm succ}^{j, \rm LOCC} = P_{\rm succ}^{\rm LOCC} \le \frac{d_{AB}}{M}.
\]    
\end{proof}

\section{Proof of Claim \ref{claim:classical_impossiblity_s}}
\begin{lemma}[The Rectangle Lemma]
For any classical communication protocol $\pi$, let $x_A \in \mathcal{X}$ and $x_B \in \mathcal{Y}$ be the inputs, $r_A \in \mathcal{R}_A$ and $r_B \in \mathcal{R}_B$ be any private random coins, and $R \in \mathcal{R}$ be public shared coins. For any fixed transcript $T$ and any fixed values of the random strings $r_A, r_B$, and $r$, the set of input pairs $(x_A, x_B)$ that results in the transcript $T$ is a combinatorial rectangle:
\[
\{(x_A, x_B) \in \mathcal{X} \times \mathcal{Y} \mid \pi(x_A, r_A, x_B, r_B, r) = T\} = S_{A, z, r_A, r} \times S_{B, z, r_B, r}
\]
where $S_{A, T, r_A, r} \subseteq \mathcal{X}$ depends only on the transcript, Alice's private coins, and the public coins, and $S_{B, T, r_B, r} \subseteq \mathcal{Y}$ depends only on the transcript, Bob's private coins, and the public coins.
\end{lemma}

\begin{lemma}
    If a classical protocol (deterministic or random) that produces a (classical) transcript $T$, if $s_a = s_b = \tilde{s} \neq \emptyset$, for any (classical) variable $R$ that is accessible to both parties $I(S; T, R) > 0$.
\end{lemma}

\begin{proof}
    In any classical communication protocol by the rectangle lemma, $x_A \perp x_B \mid T, R$, in deterministic protocols, and $(x_A, r_A) \perp (x_B, r_B) \mid T, R$, for any random variables $r_i$ held by part $i$, for any variables $R$ accessible to both parties.

    First consider the latter case, and let $(x_A, R_a, x_B, R_b)$ and $(x_A', R_a', x_B', R_b')$ be two arbitrary pairs that produce the same $(T, R)$. 
    Since Alice computes $s$ from $(x_A, R_a, T, R)$, it follows from the rectangle lemma
    $$s_a(x_A, R_a, x_B, R_b, T, R) = s_a(x_A, R_a, x_B', R_b', T, R)$$
    Since Bob computes $s$ from $(x_B', R_b', T, R)$, it follows from the rectangle lemma
    $$s_b(x_A, R_a, x_B', R_b', T, R) = s_b(x_A', R_a', x_B', R_b', T, R)$$
    Since $s_a(x_A, R_a, x_B', R_b', T, R) = s_b(x_A, R_a, x_B', R_b', T, R)$, then \\
    $s(x_A, x_B, R_a, R_b, T, R) = s(x_a, x_b', R_a, R_b', T, R) = s(x_A', x_B', R_a', R_b', T, R) = \tilde{s}$, and $H(\tilde{s} |T, R) = 0$. 
    $I(\tilde{s}; T, R) = H(\tilde{s}) - H(\tilde{s} | T, R) = H(\tilde{s}) > 0 $.
    
    The analogous argument holds for deterministic protocols.
\end{proof}
For any protocol that produces execution $\mathcal{T}_c$, the adversary's optimal estimate of the parties' recorded values $\hat{s} = (\hat{s}_a, \hat{s}_b)$, is the Maximum Likelihood estimator:
\begin{equation}
\hat{s}(\mathcal{T}_c) \coloneqq \arg\max_{(s_j, s_k) \in \mathcal{A}^2} P(s_a = s_j, s_b = s_k \mid \mathcal{T}_c).
\end{equation}

Since $s_a = s_b$
\begin{equation}
P(\hat{s} = {s} \mid \mathcal{T}_c) = \max_{s^* \in \mathcal{A}} P(s_a = s^*, s_b = s^* \mid \mathcal{T}_c).
\end{equation}

Given the previous lemma, when $s_a = s_b = \tilde{s}$, $I(\tilde{s}, \mathcal{T}_c) >0$, the posterior $P(\tilde{s} | \tilde{T}_c) \neq \frac{1}{\mathcal{A}}$, and therefore,

\begin{equation}
P(\hat{s} = {s} \mid \mathcal{T}_c) > \frac{1}{|\mathcal{A}|}.
\end{equation}

Note that in the quantum setting when $\ket{\psi}_{AB}$ is the maximally-entangled state prescribed by the primitive, the rectangle property fails. $X_a \not\perp X_b | \mathcal{T}_Q$. In particular, inputs under pairs $(x_a, x_b), (x_a', x_b')$ may yield the same $\mathcal{T}_Q$, but the corresponding cross pairs $(x_a', x_b), (x_a, x_b')$ yield \textit{different} post-unitary density matrices $\rho_{AB}$, and therefore different $\mathcal{T}_Q$. It is this property that permits the quantum setting to bypass the classical impossibility of security over $s$. 

\section{Decoupling Success Rate from Error}\label{appx:decoupling_succ_error}
By Claim \ref{claim:adversary_bounds_entangled}, and Claim \ref{claim:adversary_bounds_locc}, under either attack scenario, $P(U | \tilde{Z}) \leq P_{succ}^{PGM, j}$, for the true state $\rho^j$ corresponding to the joint unitary $U$, where PGM is performed on the $\rho_{AB}$ system, and the initial state $\ket{\psi}_{AB}$ may be \textit{any} permissible initial state (which need not be entangled). 

In the described primitive, for any true state $\rho^j$ and corresponding PGM operator $\pi_j$, $P(\rho^j | \pi_j) = P(\rho^k | \pi_j)$, for all states $\rho^k$ with corresponding unitaries $U^{(k)}$ such that $j^*(U) = j^*(U^{(k)})$. In other words, PGM measurements induce a uniform posterior over all states consistent with the correct measurement announcement $j^*(U)$. Therefore, when the adversary chooses an attack strategy consisting of performing PGM on post-unitary $\rho_{AB}$ systems, and announcing $j = j^*(\tilde{U})$, for all $\tilde{U}$ corresponding with states $\rho^i$ such that $P(\rho^i | \pi_j) > 0$, then  $P(U | \tilde{Z}) \leq \frac{1}{|\mathcal{A}|}$, while $P(s_a = s_b) = 1$.


\section{Higher-Dimensional Implementations}
We prove that one instantiation of the primitive in which the parties' unitaries are given by those prescribed below in Theorem 1, and the initial state and corresponding basis are given by either: the standard generalized Bell basis generated from \( |\Phi_d\rangle = \frac{1}{\sqrt d}\sum_{r=0}^{d-1} |r,r\rangle \) via \( |\Phi_{a,b}\rangle = (I\otimes X^b Z^a)|\Phi_d\rangle \), or the phase-twisted variant generated from \( |\Psi_d\rangle = \frac{1}{\sqrt d}\sum_{r=0}^{d-1} c_r |r,r\rangle \) with \( |c_r|=1 \), via \( |\Psi_{a,b}\rangle = (I\otimes X^b Z^a)|\Psi_d\rangle \) satisfy the security properties of the preceding claims.
\begin{lemma}[Index-shift invariance of Weyl-label ensembles]
Let $d \in \mathbb{N}$ and consider the label set
\[
\mathbb{Z}_d^4
=
\{(j_n,j_m,j_n',j_m')\}.
\]

For each
\[
j=(j_n,j_m,j_n',j_m') \in \mathbb{Z}_d^4,
\]
define the shift map induced by
\[
k=(k_n,k_m,k_n',k_m') \in \mathbb{Z}_d^4
\]
as
\[
T_k : \mathbb{Z}_d^4 \to \mathbb{Z}_d^4,
\qquad
T_k(j)
=
(j_n+k_n,\; j_m+k_m,\; j_n'+k_n',\; j_m'+k_m')
\pmod d.
\]

Then $T_k$ is a bijection (permutation) of $\mathbb{Z}_d^4$.

Consequently, for any function
\[
f : \mathbb{Z}_d^4 \to \mathbb{C},
\]
we have
\[
\sum_{j \in \mathbb{Z}_d^4} f(T_k(j))
=
\sum_{j \in \mathbb{Z}_d^4} f(j).
\]
\end{lemma}

\begin{theorem}[PGM covariance and uniform success probability]
Let
\[
U_{a,b}=X^b Z^a,
\qquad
a,b\in\mathbb Z_d,
\]
where
\[
X|r\rangle = |r+1 \!\!\!\pmod d\rangle,
\qquad
Z|r\rangle = \omega^r |r\rangle,
\qquad
\omega=e^{2\pi i/d}.
\]

For each label
\[
j=(j_n,j_m,j_n',j_m')\in\mathbb Z_d^4,
\]
define
\[
U_j
:=
U_{j_n,j_m}\otimes U_{j_n',j_m'}.
\]

Let
\[
|\psi_j\rangle := U_j|\psi\rangle,
\]
where
\[
|\psi\rangle\in\mathbb C^d\otimes\mathbb C^d
\]
is an arbitrary fixed bipartite state.

Assume the ensemble $\mathcal E = \left\{
\frac1{d^4},\,|\psi_j\rangle
\right\}_{j\in\mathbb Z_d^4}$ has a uniform prior distribution.

Let
\[
\rho
=
\frac1{d^4}
\sum_{j\in\mathbb Z_d^4}
|\psi_j\rangle\langle\psi_j|
\]
denote the average state and let
\[
\Pi_j
=
\rho^{-1/2}
\frac1{d^4}
|\psi_j\rangle\langle\psi_j|
\rho^{-1/2}
\]
denote the corresponding pretty-good measurement (PGM).

Then



The PGM success probability is independent of the true state label:
\[
\langle\psi_j|\Pi_j|\psi_j\rangle
=
\langle\psi_0|\Pi_0|\psi_0\rangle
\qquad
\forall\,j\in\mathbb Z_d^4.
\]

\end{theorem}

\begin{proof}
Fix
\[
k=(k_n,k_m,k_n',k_m')
\in\mathbb Z_d^4.
\]

Using
\[
Z^aX^b=\omega^{ab}X^bZ^a,
\]
we obtain
\[
U_{k_n,k_m}U_{j_n,j_m}
=
(X^{k_m}Z^{k_n})
(X^{j_m}Z^{j_n})
=
\omega^{k_nj_m}
U_{k_n+j_n,\;k_m+j_m},
\]
where all additions are modulo $d$.

Similarly,
\[
U_{k_n',k_m'}U_{j_n',j_m'}
=
\omega^{k_n'j_m'}
U_{k_n'+j_n',\;k_m'+j_m'}.
\]

Therefore
\[
U_kU_j
=
\omega^{k_nj_m+k_n'j_m'}
U_{k\oplus j},
\]
where
\[
k\oplus j
=
(k_n+j_n,\,
k_m+j_m,\,
k_n'+j_n',\,
k_m'+j_m')
\pmod d.
\]

Hence
\[
U_k|\psi_j\rangle
=
U_kU_j|\psi\rangle
=
\omega^{k_nj_m+k_n'j_m'}
|\psi_{k\oplus j}\rangle.
\]

Since the phase cancels in the corresponding projector,
\[
U_k
|\psi_j\rangle\langle\psi_j|
U_k^\dagger
=
|\psi_{k\oplus j}\rangle
\langle\psi_{k\oplus j}|.
\]

Consequently,
\[
U_k\rho U_k^\dagger
=
\frac1{d^4}
\sum_{j\in\mathbb Z_d^4}
|\psi_{k\oplus j}\rangle
\langle\psi_{k\oplus j}|.
\]

By the previous lemma, applied componentwise to $\mathbb Z_d^4$, the map $j\mapsto k\oplus j$ is a permutation of the label set. Therefore
\[
\sum_{j\in\mathbb Z_d^4}
|\psi_{k\oplus j}\rangle
\langle\psi_{k\oplus j}|
=
\sum_{j\in\mathbb Z_d^4}
|\psi_j\rangle
\langle\psi_j|,
\]
and hence
\[
U_k\rho U_k^\dagger
=
\rho.
\]

Therefore
\[
[\rho,U_k]=0
\qquad
\forall\,k\in\mathbb Z_d^4.
\]

Since $\rho$ commutes with every $U_k$, every spectral function of $\rho$
also commutes with every $U_k$. In particular,
\[
[\rho^{-1/2},U_k]=0.
\]

Using
\[
|\psi_j\rangle
=
U_j|\psi_0\rangle,
\]
we obtain
\[
\begin{aligned}
\Pi_j
&=
\rho^{-1/2}
\frac1{d^4}
U_j
|\psi_0\rangle
\langle\psi_0|
U_j^\dagger
\rho^{-1/2}
\\
&=
U_j
\left(
\rho^{-1/2}
\frac1{d^4}
|\psi_0\rangle
\langle\psi_0|
\rho^{-1/2}
\right)
U_j^\dagger.
\end{aligned}
\]

Defining
\[
\Pi_0
=
\rho^{-1/2}
\frac1{d^4}
|\psi_0\rangle
\langle\psi_0|
\rho^{-1/2},
\]
it follows that
\[
\Pi_j
=
U_j\Pi_0U_j^\dagger.
\]

Finally,
\[
\begin{aligned}
\langle\psi_j|\Pi_j|\psi_j\rangle
&=
\langle\psi_0|
U_j^\dagger
(U_j\Pi_0U_j^\dagger)
U_j
|\psi_0\rangle
\\
&=
\langle\psi_0|
\Pi_0
|\psi_0\rangle.
\end{aligned}
\]

Therefore
\[
\langle\psi_j|\Pi_j|\psi_j\rangle
=
\langle\psi_0|\Pi_0|\psi_0\rangle
\]
for all
\[
j\in\mathbb Z_d^4.
\]

Thus the PGM success probability is identical for every state labeled by $j$.
\end{proof}
By the above, and by Proof of \ref{claim:adversary_bounds_entangled}, under the unitary operators prescribed above, since there are $d^4$ states in $d^2$ dimensions, $P_{succ}^{PGM, j} \leq \frac{d^2}{d^4} = \frac{1}{d^2}$, and hence any recording strategies for functions with $|\mathcal{A}| \leq \frac{1}{d^2}$, abiding by the prescribed properties, will satisfy the security constraints given by Claim \ref{claim:adversary_bounds_entangled} and Claim \ref{claim:adversary_bounds_locc}.

\end{document}